\documentclass{tlp}
\newtheorem{lemma}{Lemma}
\newtheorem{theorem}{Theorem}

\def\G{\Gamma}
\def\D{\Delta}
\def\seq{\Rightarrow}
\def\r#1#2{\frac{\textstyle #1}{\textstyle #2}}

\def\HT{\emph{HT\/}}
\def\HTA{\emph{HTA\/}}
\def\HTAD{\emph{HTA}\!+\!\emph{Defs}}
\def\HTD{\emph{HT\/}^\infty\!+\!gr(\emph{Defs\/})}
\def\co{\emph{count\/}}
\def\atl{\emph{Atleast\/}}
\def\atm{\emph{Atmost\/}}

\def\num{\overline}
\def\head{\emph{Head}}
\def\body{\emph{Body}}
\def\ar{\leftarrow}
\def\lrar{\leftrightarrow}
\def\no{\emph{not}}

\def\clingo{{\sc clingo}}
\def\gringo{{\sc gringo}}

\def\beq{\begin{equation}}
\def\eeq#1{\label{#1}\end{equation}}
\def\ba{\begin{array}}
\def\ea{\end{array}}

\def\val#1#2{\emph{val\,}_{#1}({#2})}

\def\p2f{\hbox{p2f}}

\def\bft{{\bf t}}
\def\bfx{{\bf X}}
\def\bfv{{\bf V}}
\def\bfw{{\bf W}}

\begin{document}

\lefttitle{Vladimir Lifschitz}

\jnlPage{1}{16}

\begin{authgrp}
\title[Strong Equivalence of Logic Programs with Counting]
{Strong Equivalence\\ of Logic Programs with Counting}
  \author{Vladimir Lifschitz}
  \affiliation{University of Texas at Austin, USA}
\end{authgrp}

\maketitle

\begin{abstract}
In answer set programming, two groups of rules are considered strongly
equivalent if they have the same meaning in any context.
In some cases, strong equivalence of programs in the
input language of the grounder \gringo\ can be established
by deriving rules of each
program from rules of the other.  The possibility of such proofs has been
demonstrated for a subset of that language  that includes comparisons,
arithmetic operations, and simple choice rules, but not aggregates.
This method is extended here to a class of programs
in which some uses of the \verb|#count| aggregate are allowed.
This paper is under consideration for acceptance in TPLP.
\end{abstract}

\section{Introduction}

In answer set programming \cite{mar99,nie99,gel14,lif19a},
two groups of rules are considered strongly
equivalent if, informally speaking, they have the same meaning in any context
\cite{lif01}.
We are interested in proving strong equivalence of programs in the
input language of the grounder \gringo\ \cite{gringomanual}
by deriving rules of each
program from rules of the other.  The possibility of such proofs has been
demonstrated for the subset of that language called mini-\gringo\
\cite{lif19,lif21a}.  Programs allowed in that subset may include comparisons,
arithmetic operations, and simple choice rules, but not aggregates.

The process of proving strong equivalence uses the translation~$\tau^*$
\cite[Section~6]{lif19} that transforms mini-\gringo\
rules into first-order formulas with two sorts of variables---for
numerals and for arbitrary precomputed terms. 
If two mini-\gringo\ programs, rewritten as
sets of sentences, are equivalent in the deductive system \HTA\
(``here-and-there with arithmetic'')
then they are strongly equivalent \cite[Section~4]{lif21a}.

In this paper,~$\tau^*$ is extended to a superset of mini-\gringo\
in which the
\verb|#count| aggregate can be used in a limited way.  The study of the
strong equivalence relation between \gringo\ programs with counting and
other aggregates is important because these constructs are
widely used in answer set programming, and because some properties
of this relation in the presence of aggregates
may seem counterintuitive.  For instance, the rule
\beq
\verb|q :- #count{X : p(X)} >= Y, Y = 1.|
\eeq{g1}
is strongly equivalent to each of the simpler rules
\beq
\verb|q :- #count{X : p(X)} >= 1.|
\eeq{g3}
and
\beq
\verb|q :- p(X).|
\eeq{g3a}
---as could be expected.  But the rule
\beq
\verb|q :- #count{X : p(X)} = Y, Y >= 1.|
\eeq{g2}
is not strongly equivalent to~(\ref{g3}) and~(\ref{g3a}).  Indeed, adding
the rules
\beq\ba l
\verb|p(a).|\\
\verb|p(b) :- q.|
\ea\eeq{bad}
to~(\ref{g2}) gives a program without stable
models.\footnote{\label{ft1} This is a modification of an example due to
  Gelfond and Zhang \citeyear{gel19}.  This example shows that the use of
  functional notation for {\tt \#count} is sometimes misleading. The
  correspondence between sets and their cardinalities is a total function
  classically, but not intuitionistically.  We return to this question in
the discussion of related work (Section~\ref{sec:rw}).}

The syntax and semantics of mini-\gringo\ rules
are reviewed in Sections~\ref{sec:synmg} and~\ref{sec:semmg}.
The semantics is defined by transforming rules
into infinite sets of propositional formulas
\cite[Section~3]{lif19}
and then appealing to the propositional stable model semantics
\cite[Section~5.2]{lif19a}.  Adding the \verb|#count| aggregate,
described in Sections~\ref{sec:synmgc} and~\ref{sec:semmgc}, involves
stable models of 
\emph{infinitary} propositional formulas \cite[Section~2]{tru12},
in the spirit of the approach of Gebser et al.~\citeyear{geb15}.
The translation~$\tau^*$ is reviewed in Section~\ref{sec:formulas} and
extended to mini-\gringo\ with counting in
Sections~\ref{sec:signature},~\ref{sec:trans}.
After providing additional background information in Section~\ref{sec:infht},
we state in Section~\ref{sec:properties} three theorems expressing
properties of the translation and show how the strong equivalence
of rules~(\ref{g1})--(\ref{g3a}) can be proved by deriving them from each
other.  Proofs of the theorems are given in Section~\ref{sec:proofs}.

\section{Review: syntax of mini-gringo} \label{sec:synmg}

The description of mini-\gringo\ programs below uses ``abstract syntax,''
which disregards some details related to representing programs by
strings of ASCII characters.
We assume that three countably infinite sets of symbols are selected:
\emph{numerals}, \emph{symbolic constants}, and \emph{variables}.
We assume that a 1-1 correspondence between numerals
and integers is chosen; the numeral corresponding to an integer~$n$ is
denoted by $\num n$.
\emph{Precomputed terms} are numerals, symbolic constants, and the
symbols \emph{inf}, \emph{sup}.  We assume that a total order on the set
of precomputed terms is selected, with the least element \emph{inf} and the
greatest element \emph{sup}, so that numerals are contiguous and ordered in
the standard way.

  Terms allowed in a mini-\gringo\ program are formed from
precomputed terms and variables using the binary operation symbols
$$+\quad-\quad\times\quad/\quad\backslash \quad..$$
An \emph{atom} is a symbolic constant optionally followed by a tuple
of terms in parentheses.  A \emph{literal} is an atom
possibly preceded by one or two occurrences of \emph{not}. A \emph{comparison}
is an expression of the form
$t_1\prec t_2$, where $t_1$, $t_2$ are mini-{\sc gringo} terms and $\prec$ is one
of the six comparison symbols
\beq
=\quad\neq\quad<\quad>\quad\leq\quad\geq
\eeq{comp}

A \emph{mini-\gringo\ rule} is an expression of the form
\beq
\head\ar\body,
\eeq{rule}
where
\begin{itemize}
\item
\strut$\body$ is a conjunction (possibly empty) of literals and comparisons,
and
\item
  $\head$
  is either an atom (then~(\ref{rule}) is a
  \emph{basic rule\/}), or an atom in braces
(then~(\ref{rule}) is a \emph{choice rule}\/), or empty (then~(\ref{rule})
is a \emph{constraint}).
\end{itemize}
A \emph{mini-\gringo\ program} is a finite set of mini-\gringo\ rules.

\section{Review: semantics of mini-gringo} \label{sec:semmg}

The semantics of ground terms is defined by assigning to every ground term~$t$
the finite set~$[t]$ of its \emph{values} \cite[Section~3]{lif19}.  Values
of a ground term are precomputed terms.  For instance,
$$
[\num 2/\num 2]=\{\num 1\},\
[\num 2/\num 0]=\emptyset,\
[\num 0\,..\,\num 2]=\{\num 0, \num 1, \num 2\}.
$$
If a term is interval-free (that is, does not contain $..$) then it has at
most one value.
For any ground terms $t_1, \dots, t_n$, by $[t_1, \dots, t_n]$ 
we denote the set of tuples $r_1, \dots, r_n$ such that
$r_1 \in [t_1], \dots,$ $r_n \in [t_n]$.

Stable models of a mini-\gringo\ program
are defined as stable models of the set of propositional formulas obtained
from it by applying a syntactic transformation denoted by~$\tau$
\cite[Section~3]{lif19}.
These propositional formulas are built from \emph{precomputed atoms\/}---atoms
$p({\bf r})$ such that members of the tuple~$\bf r$ are precomputed terms.
Thus every stable model is a set of precomputed atoms.

The transformation~$\tau$ is defined as follows.
For any ground atom $p({\bf t})$,
\begin{itemize}
\item
$\tau(p({\bf t}))$ is
$\bigvee_{{\bf r} \in [{\bf t}]} p({\bf r})$, 
\item
$\tau (\no\ p({\bf t}))$ is 
$\bigvee_{{\bf r} \in [{\bf t}]} \neg p({\bf r})$, and
\item
 $\tau (\no\ \no\ p({\bf t}))$ is
$\bigvee_{{\bf r} \in [{\bf t}]} \neg\neg p({\bf r})$.
\end{itemize}
For any ground comparison $t_1 \prec t_2$, $\tau(t_1 \prec t_2)$ is
$\top$ if the relation~$\prec$ holds between some~$r_1$ from $[t_1]$ and
some~$r_2$ from~$[t_2]$, and $\bot$ otherwise.
The result of applying~$\tau$ to a conjunction $B_1\land B_2\land \cdots$
is $\tau(B_1)\land \tau(B_2)\land\cdots$.
If~$R$ is a ground basic rule $p({\bf t}) \ar \body$
then $\tau(R)$ is the propositional formula
\beq
\tau(\body) \to
\bigwedge_{{\bf r}\in [{\bf t}]} p({\bf r}).
\eeq{deftau1}
If~$R$ is a ground choice rule $\{p({\bf t})\} \ar \body$
then $\tau(R)$ is the propositional formula
\beq
\tau(\body) \to \bigwedge_{{\bf r} \in [{\bf t}]}
(p({\bf r})\lor\neg p({\bf r})).
\eeq{deftau2}
If~$R$ is a ground constraint $\ar\body$ then $\tau(R)$ is
\beq
\neg\tau(\body).
\eeq{deftau3}
For any mini-\gringo\ program~$\Pi$, $\tau(\Pi)$ is the set of
propositional formulas~$\tau(R)$ for all ground rules~$R$ that can be
obtained from rules of~$\Pi$ by substituting precomputed terms for variables.

For example, substituting a precomputed term~$r$ for~$X$ in the choice rule
\beq
\{p(\num 2)\} \ar p(X)
\eeq{example}
gives the ground rule $\{p(\num 2)\} \ar p(r)$, so that~$\tau$
transforms~(\ref{example}) into the set of all formulas of the form
$$p(r) \to p(\num 2) \lor \neg p(\num 2).$$

\section{Mini-gringo with counting: syntax}  \label{sec:synmgc}

We extend the class of mini-\gringo\ rules as follows.
An \emph{aggregate element} is a pair
$\bfx:{\bf L},$
where $\bfx$ is a tuple of distinct variables,
and $\bf L$ is a conjunction of
literals and comparisons such that every member of~$\bfx$ occurs in~$\bf L$.
In mini-\gringo\ with counting, the body of a rule is allowed
to contain, besides literals and comparisons, \emph{aggregate atoms} of the
forms
\beq
\co \{E\} \geq t,\  \co \{E\} \leq t,\ 
\eeq{aa}
where~$E$ is an aggregate element, and~$t$ is an interval-free term.
The conjunction of aggregate atoms~(\ref{aa}) can be written as
$\co\{E\}=t$.

A variable that occurs in a rule~$R$ is \emph{local} in~$R$ if each of
its occurrences is within an aggregate element, and \emph{global} otherwise.
A rule is \emph{pure} if, for every aggregate element $\bfx:{\bf L}$ in its
body, all variables in the tuple~$\bfx$ are local.  For example, all
rules that do not contain aggregate elements are pure.  The rules
\beq
q \ar \co\{X:p(X,Y)\} \leq \num 2
\eeq{pr1}
and
\beq
q \ar \co\{X:p(X,Y)\} \leq \num 2\land Y=\num 1\,..\,\num {10},
\eeq{pr0}
are pure, because~$X$ is local in each of them.  The rule
\beq
q \ar \co\{X:p(X,Y)\} \leq \num 2\land X=\num 1\,..\,\num {10}
\eeq{notsafe}
is not pure, because~$X$ is global.\footnote{In response to rules like this,
  the current version of \gringo\ produces a warning message: {\tt global
    variable in tuple of aggregate element.}}

A \emph{program in mini-\gringo\ with counting}, or an
\emph{{\sc mgc} program}, is a finite set of pure rules.  Allowing non-pure
rules in a program would necessitate making the semantics more complicated;
see Footnote~\ref{ft2}.

An expression of the form
$$\num m\,\{{\bf X}:A\}\,\num n \ar\body$$
where~$\bf X$ is a tuple of distinct variables, $A$ is an atom, and
$\body$ is a conjunction of literals, comparisons and aggregate atoms,
can be used as shorthand for the group of three rules:
$$\ba l
\{A\} \ar\body,\\
\ar \body, \co\{{\bf X}:A\} \leq \num{m-1},\\
\ar \body, \co\{{\bf X}:A\} \geq \num{n+1}.
\ea$$

\section{Mini-gringo with counting: semantics} \label{sec:semmgc}

To define the semantics of {\sc mgc} programs, we will
extend the definition of~$\tau$
(Section~\ref{sec:semmg}) to aggregate
atoms~(\ref{aa}) such that~$t$ is a ground term.  Since~$t$ is
interval-free,~$[t]$ is either a singleton or empty.  If $[t]$ is
a singleton $\{c\}$ then $\tau(\co\{\bfx:{\bf L}\} \geq t)$ is defined as
the infinite disjunction
\beq
\bigvee_{\Delta\,:\,\num{|\Delta|} \geq c}\;\bigwedge_{{\bf x}\in\Delta}\;\bigvee_{\bf w}
   \tau\!\left({\bf L}^{\bfx,\bfw}_{\,{\bf x},\;{\bf w}}\right),
\eeq{inf1}
and $\tau(\co\{\bfx:{\bf L}\} \leq t)$ as the infinite conjunction
\beq
\bigwedge_{\Delta\,:\,\num{|\Delta|}>c}\;
\neg\bigwedge_{{\bf x}\in\Delta}\;\bigvee_{\bf w}
   \tau\!\left({\bf L}^{\bfx,\bfw}_{\,{\bf x},\;{\bf w}}\right),
\eeq{inf2}
where
\begin{itemize}
\item $\Delta$ ranges over finite sets of tuples of precomputed terms of
  the same length as~$\bfx$;
\item
  $\bfw$ is the list of variables that occur in $\bf L$ but do not belong
  to~$\bfx$;
  \item $\bf w$ ranges over tuples of precomputed terms of the same
    length as $\bfw$;
  \item the expression
  ${\bf L}^{\bfx,\bfw}_{\,{\bf x},\;{\bf w}}$ denotes the result of
  substituting $\bf x$, $\bf w$ for all occurrences of $\bfx$, $\bfw$
  in~{\bf L}.
  \end{itemize}
  (If~$c$ is a numeral~$\num n$ then the inequalities $\num{|\Delta|}\geq c$,
  $\num{|\Delta|}> c$ in these formulas
  can be written as $|\Delta|\geq n$, $|\Delta|>n$.)
If $[t]$ is empty then we define
$$\tau(\co\{\bfx:{\bf L}\} \geq t)=\tau(\co\{\bfx:{\bf L}\} \leq t)=\bot.$$

For example, the result of applying~$\tau$ to the body of rule~(\ref{pr1}) is
$$\bigwedge_{\Delta\,:\,|\Delta| > \num 2}\neg\bigwedge_{x\in\Delta}\;
\bigvee_w p(x,w).$$
(In this case, $\bfx$ is~$X$; $\bfw$ is~$Y$;
$\Delta$ ranges over sets of precomputed terms;
$w$ ranges over precomputed terms.)  In application to the aggregate
expression
$$\co\{X:p(X,r)\} \leq \num 2,$$
where~$r$ is a precomputed term,~$\tau$ gives
$$\bigwedge_{\Delta\,:\,|\Delta|>\num 2}\neg\bigwedge_{x\in\Delta}p(x,r)$$
($\bfw$ is empty).

A rule is \emph{closed} if it has no global variables.
It is clear that substituting precomputed terms for all global
variables in a pure rule is a closed pure rule.\footnote{\label{ft2} In
  case of a rule that is not pure, substituting precomputed terms for
  global variables may transform an aggregate element in the body
  into an expression that is not an aggregate element.  For instance,
  substituting $\num 1$ for~$X$ in rule~(\ref{notsafe}) turns
  $X:p(X,Y)$ into the expression
  $\num 1:p(\num 1,Y)$, which is not allowed by the syntax of
  mini-\gringo\ with counting.}  For any mini-\gringo\
program~$\Pi$, $\tau(\Pi)$ stands for the conjunction of the
formulas~$\tau(R)$ over all closed rules~$R$ that can be
obtained from rules of~$\Pi$ by such substitutions.  Thus~$\tau(\Pi)$ is
an infinitary propositional formula over the signature consisting of
all precomputed atoms.

For example,~$\tau$ transforms rule~(\ref{pr1}) into
$$\bigwedge_{\Delta\,:\,|\Delta| > \num 2}\neg\bigwedge_{x\in\Delta}\;
\bigvee_w p(x,w) \to q,$$
where $\Delta$ ranges over sets of precomputed terms, and
$w$ ranges over precomputed terms.  Rule~(\ref{pr0}) becomes
$$\bigwedge_r\left(\left(\bigwedge_{\Delta\,:\,|\Delta|>\num 2}
  \neg\bigwedge_{x\in\Delta}p(x,r)\right) \land \tau(r=\num 1..\num{10})
\to q\right),$$
where~$r$ ranges over precomputed terms, and~$\Delta$ ranges over sets of
precomputed terms.
The subformula $\tau(r=\num 1..\num{10})$ is~$\top$ if~$r$ is one of the
numerals $\num 1,\dots,\num{10}$, and~$\bot$ otherwise.

The semantics of {\sc mgc} described in this section, like the semantics of
aggregates proposed by Gebser et al.~\citeyear{geb15}, aims at modeling the
behavior of the answer set solver {\sc clingo}.  The definition in this
paper is simpler than the 2015 version, but more limited in scope,
because it is does not cover aggregates other than \verb|#count|.  The two
versions of the semantics of \verb|#count| are not completely equivalent,
however, because they
handle infinite sets in slightly different ways.  This difference does not
affect safe programs, and is in this sense inessential.

\section{Review: representing mini-gringo rules by formulas}
\label{sec:formulas}

The target language of the translation~$\tau^*$ is a first-order language
with two sorts: the sort \emph{general} and its subsort
\emph{integer}.  General variables are meant to
range over arbitrary precomputed terms, and we identify them with
variables used in mini-\gringo\ rules. Integer of the second sort are meant
to range over numerals (or, equivalently, integers).
The signature~$\sigma_0$ of the language includes
\begin{itemize}
\item all precomputed terms as object constants; an object constant
  is assigned the sort \emph{integer} iff it is a numeral;
\item the symbols~$+$, $-$ and~$\times$ as binary function constants;
  their arguments and values have the sort \emph{integer};
\item symbols $p/n$, where~$p$ is a symbolic constant,
  as $n$-ary predicate constants;
\item comparison symbols~(\ref{comp}) as binary predicate constants.
\end{itemize}
An atomic formula $(p/n)({\bf t})$ can be abbreviated as
$p({\bf t})$. An atomic formula $\prec\!\!(t_1,t_2)$, where~$\prec$ is a
comparison symbol, can be written as $t_1\prec t_2$.

Lifschitz et al.~\citeyear{lif19} defined, for every mini-{\sc gringo}
term~$t$, a formula $\val tZ$ that expresses,
informally speaking, that~$Z$ is one of the values of~$t$.  For example,
$\val {\num 2}Z$ is $Z=\num 2$.
If $\bf t$ is a tuple $t_1,\dots,t_n$ of mini-{\sc gringo} terms, and
$\bf Z$ is a tuple $Z_1,\dots,Z_n$ of distinct general variables, then
$\val{\bf t}{\bf Z}$ stands for
$\val{t_1}{Z_1} \land \cdots \land \val{t_n}{Z_n}$.

The translation~$\tau^B$, which transforms literals
and comparisons into formulas over the signature~$\sigma_0$,
is defined in that paper as
follows:\footnote{The superscript~$B$ indicates that
this translation is intended for \emph{bodies} of rules.}
\begin{itemize}
\item
  $\tau^B(p({\bf t}))=
  \exists {\bf Z}(\val{\bf t}{\bf Z} \land p({\bf Z}))$;
\item
  $\tau^B(\no\ p({\bf t})) =
  \exists {\bf Z}(\val{\bf t}{\bf Z} \land \neg p({\bf Z}))$;
\item
  $\tau^B(\no\ \no\ p({\bf t})) =
  \exists {\bf Z}(\val{\bf t}{\bf Z} \land \neg\neg p({\bf Z}))$;
\item
$\tau^B(t_1\prec t_2) =
\exists Z_1 Z_2 (\val{t_1}{Z_1} \land \val{t_2}{Z_2} \land
Z_1\prec Z_2)$.
\end{itemize}
Here $Z_1$,~$Z_2$, and members of the tuple {\bf Z} are fresh general
variables.

The result of applying $\tau^*$ to a mini-\gringo\ rule
$H \ar B_1\land\cdots\land B_n$
can be defined as the universal closure of the formula
\beq\ba{ll}
B^*_1\land\cdots\land B^*_n\land\val{\bf t}{\bf Z}\to p({\bf Z})
&\hbox{ if }H\hbox{ is }p({\bf t}),\\
B^*_1\land\cdots\land B^*_n\land\val{\bf t}{\bf Z}
  \to p({\bf Z}) \lor \neg p({\bf Z})
&\hbox{ if }H\hbox{ is }p\{({\bf t})\},\\
\neg(B^*_1\land\cdots\land B^*_n)
&\hbox{ if $H$ is empty},
\ea\eeq{bstar}
where $B^*_i$ stands for $\tau^B(B_i)$, and~$\bf Z$ is a tuple of
fresh general variables.

For example, the
result of applying~$\tau^*$ to choice rule~(\ref{example}) is
$$\forall XZ(\tau^B(p(X))\land Z=\num 2 \to p(Z)\lor \neg p(Z));$$
$\tau^B(p(X))$ can be further expanded into $\exists Z(Z=X \land p(Z))$.

\section{Extending the target language}\label{sec:signature}

In this section, we extend the translation~$\tau^*$ to arbitrary pure rules.
This more general translation produces first-order formulas over the
signature~$\sigma_1$ that is obtained from~$\sigma_0$ by adding infinitely
many predicate constants
\beq
\atl^{\bfx;\bfv}_F\hbox{ and }\atm^{\bfx;\bfv}_F
\eeq{newpc}
where $\bfx$ and $\bfv$
are disjoint lists of distinct general variables, and~$F$ is a formula
over~$\sigma_0$ such that each of its free variables belongs to $\bfx$
or to~$\bfv$.\footnote{Adding infinitely many predicate symbols gives us a
  single signature that is sufficient for representing all {\sc mgc}
  rules.  The translation of any specific rule will only contain finitely many
  symbols, of course.}
The number of arguments of each of constants~(\ref{newpc})
is greater by~1 than the
length of~$\bfv$; all arguments are of the sort \emph{general}.

If~$n$ is a positive integer then the
formula $\atl^{\bfx,\bfv}_F(\bfv,\num n)$ is meant to express
that~$F$ holds for at least~$n$ values of~$\bfx$; symbolically,
\beq
\exists\bfx_1\cdots\bfx_n\left(\,\bigwedge_{i=1}^n F^\bfx_{\,\bfx_i}
  \land \bigwedge_{i<j}\neg(\bfx_i=\bfx_j)  \right),
\eeq{atln}
where $\bfx_1,\dots,\bfx_n$ are tuples of fresh general
variables.\footnote{An expression of the form
  $(X_1,X_2,\dots)=(Y_1,Y_2,\dots)$ stands for
  $X_1=Y_1\land X_2=Y_2\land\cdots$.}
For any precomputed term~$r$, the expression $\exists_{\geq r}\bfx F$
will stand for
\medskip

\begin{tabular}{ll}
  formula~(\ref{atln}),     & if $r=\num n>\num 0$,\\
  $\top$,                   & if $r \leq \num 0$,\\
  $\bot$,                   & if $r>\num n$ for all integers~$n$.
\end{tabular}
\medskip

The formula $\atm^{\bfx,\bfv}_F(\bfv,\num n)$ is meant to express
that~$F$ holds for at most~$n$ values of~$\bfx$; symbolically,
\beq
\forall\bfx_1\cdots\bfx_{n+1}\left(\bigwedge_{i=1}^{n+1}
  F^\bfx_{\,\bfx_i}
\to \,\bigvee_{i<j}\bfx_i=\bfx_j\right).
\eeq{atmn}
For any precomputed term~$r$, the expression
$\exists_{\leq r}\bfx F$ will stand for
\medskip

\begin{tabular}{ll}
  formula~(\ref{atmn}),     & if $r=\num n\geq \num 0$,\\
  $\bot$,                   & if $r < \num 0$,\\
  $\top$,                   & if $r>\num n$ for all integers~$n$.
\end{tabular}
\medskip

The set of all sentences of the forms
\beq
\forall \bfv \left(\atl^{\bfx;\bfv}_F(\bfv,r)\lrar
  \exists_{\geq r} \bfx F\right),
\eeq{def1}
\beq
\forall \bfv \left(\atm^{\bfx;\bfv}_F(\bfv,r)\lrar
  \exists_{\leq r}\bfx F\right)
\eeq{def2}
will be denoted by \emph{Defs}.

\section{Representing pure rules by formulas}\label{sec:trans}

To extend the definition of~$\tau^*$ reproduced in Section~\ref{sec:formulas}
to arbitrary pure rules, we need to say how to choose $B^*_i$ in~(\ref{bstar})
when $B_i$ includes an aggregate element $\bfx:{\bf L}$.  Let~$\bfv$ be
the list of global variables that occur in $\bf L$, and let~$\bfw$ be the
list of local variables that occur in~$\bf L$ but are not
included in~$\bfx$.  Then~$B^*_i$ is defined as
$$
\exists C\left(\val t C
\land \atl^{\bfx;\bfv}_{\exists \bfw\tau^B({\bf L})}(\bfv,C)\right)
$$
if $B_i$ is $\co\{\bfx:{\bf L}\} \geq t$, and as
$$
\exists C\left(\val t C
\land \atm^{\bfx;\bfv}_{\exists \bfw\tau^B({\bf L})}(\bfv,C)\right)
$$
if $B_i$ is $\co\{\bfx:{\bf L}\} \leq t$,
where~$C$ is a fresh general variable.

For example, the result of applying~$\tau^*$ to rule~(\ref{pr1}) is
$$
\exists C
\left(\val {\num 2} C \land \atm^{X;}_{\exists Y\tau^B(p(X,Y))}(C)\right)
\to q
$$
($\bfv$ is empty, $\bfw$ is $Y$).
The result of applying~$\tau^*$ to rule~(\ref{pr0}) is
$$
\forall Y\left(\exists C\left(\val {\num 2} C \land\atm^{X;Y}_{\tau^B(p(X,Y))}(Y,C)\right)\land
\tau^B(Y=\num 1\,..\,\num {10}) \to q\right)
$$
($\bfv$ is~$Y$, $\bfw$ is empty).

For any {\sc mgc}
program~$\Pi$, $\tau^*(\Pi)$ stands for the conjunction of the formulas
$\tau^*(R)$ for all rules~$R$ of~$\Pi$.  Thus~$\tau^*(\Pi)$ is a
sentence over the signature~$\sigma_1$.

\section{Review: infinitary logic of here-and-there} \label{sec:infht}

Some of the properties of the translation~$\tau^*$ discussed below refer to
the deductive system of infinitary propositional logic of here-and-there
\cite{har17}, denoted by~$\HT^\infty$.  In this section, we reproduce
the definition of that system.

The derivable objects of $\HT^{\infty}$ are \emph{sequents\/}---expressions
of the form $\G\seq F$, where~$F$ is an infinitary propositional
formula, and 
$\G$ is a finite set of infinitary propositional
formulas (``$F$ under assumptions~$\G$'').
To simplify notation, we write $\G$ as a list.
We identify a sequent of the form $\seq F$ with the formula~$F$.

The axiom schemas of $\HT^\infty$ are
$$F \seq F,$$
$$
F\lor(F\to G)\lor\neg G
$$
and
\beq
\bigwedge_{\alpha \in A}\ \ \bigvee_{F \in \mathcal{H}_\alpha} F \to 
\bigvee_{(F_\alpha)_{\alpha \in A}}\ \  \bigwedge_{\alpha \in A} F_\alpha,
\eeq{dist_cod}
where $(\mathcal{H}_\alpha)_{\alpha \in A}$ is a non-empty family of sets
of formulas; the disjunction in the
consequent of (\ref{dist_cod}) extends over all elements
$(F_\alpha)_{\alpha \in A}$ of the Cartesian product of the family
$(\mathcal{H}_\alpha)_{\alpha \in A}$.
The inference rules of $\HT^\infty$ are the introduction
and elimination rules for the propositional connectives
shown in the table above and the weakening rule
$$
(W)\;\r {\G\seq F}{\G,\D\seq F}.
$$
\begin{table}
\begin{center}
\begin{tabular}{ll}
\hline
\\
$
\!(\land I)\;\r	{\G\seq H \quad\hbox{for all }H\in\mathcal H}
		{\G\seq \mathcal H^\land}$
&$\quad		
(\land E)\;\r	{\G\seq \mathcal H^\land}
		{\G\seq H}
	\quad(H\in\mathcal H)
$
\\ \\
$\!(\lor I)\;\r	{\G\seq H}
		{\G\seq \mathcal H^\lor}
\quad(H\in\mathcal H)$
&\quad				
$(\lor E)\;\r{\G\seq \mathcal H^\lor \qquad \D,H \seq F
    \quad\hbox{for all }H\in\mathcal H}
	    {\G,\D\seq F}$
\\ \\
$\!(\to\!\! I)\;\r	{\G,F\seq G}
			{\G\seq F\to G}$
&$\quad				
(\to\!\! E)\;\r	{\G\seq F\quad \D\seq F \to G}
			{\G,\D\seq G}$\\
\\
\hline
\end{tabular}
\end{center}
\caption{Introduction and elimination rules of infinitary propositional
  logic. By~$\mathcal H^\land$ and
$\mathcal H^\lor$ we denote the conjunction and disjunction of all
formulas in~$\mathcal H$.}
\end{table}
Falsity and   negation are not mentioned in the axiom schemas and
inference rules of~$\HT^\infty$
  because~$\bot$ is considered shorthand for $\emptyset^\lor$, and
  $\neg F$ is shorthand for $F\to\bot$.

The set of \emph{theorems of HT\/$^\infty$} is the smallest set of sequents
that includes the axioms of the system and is closed under the application
of its inference rules.
We say that  formulas~$F$ and~$G$ are \emph{equivalent in HT\/$^\infty$} if 
$F\lrar G$ is a theorem of~\HT$^\infty$. 

The role of this deductive system is determined by the fact that
two infinitary propositional formulas are strongly equivalent to each
other if and only if they are equivalent in $\HT^\infty$
\cite[Corollary~2]{har17}.

\section{Properties of the generalized translation} \label{sec:properties}

Informally speaking, a pure rule~$R$ has the same meaning as the
sentence~$\tau^*(R)$.
This claim is made precise in Theorem~\ref{thm1} below.  The statement
of the theorem refers to the infinitary propositional formulas
obtained from sentences over~$\sigma_1$ by applying
the grounding operator $gr$, which is defined recursively:
\begin{itemize}
\item 
  $gr(\bot)$ is $\bot$;
\item
  if $F$ is $\prec(t_1,t_2)$, where~$\prec$ is a comparison symbol,
  then $gr(F)$ is $\top$ if the relation~$\prec$ holds for the values
  of $t_1$ and $t_2$, and $\bot$ otherwise;
\item
  if $F$ is $p({\bf t})$, where~$p$ is not a comparison symbol,
  then $gr(F)$ is obtained from~$F$ by replacing
  each member of the tuple~${\bf t}$ by its value;
\item
  $gr(F\odot G)$ is $gr(F)\odot gr(G)$ for every binary connective~$\odot$;
\item
  $gr(\forall X\,F)$ is the conjunction of the formulas
  $gr\left(F^X_r\right)$
  over all precomputed terms $r$ if $X$ is a general variable,
and over all numerals~$r$ if $X$ is an integer variable;
\item
  $gr(\exists X\,F)$ is the disjunction of the formulas
  $gr\left(F^X_r\right)$
  over all precomputed terms $r$ if $X$ is a general variable,
and over all numerals~$r$ if $X$ is an integer variable.
\end{itemize}
Thus $gr(F)$ is an infinitary propositional formula over the
signature consisting of all atomic formulas of the form $p({\bf r})$,
where~$p$ is different from comparison symbols and $\bf r$ is a tuple
of precomputed terms.  Such atomic formulas will be called \emph{extended
precomputed atoms}.  Unlike precomputed atoms, they may contain
predicate symbols~(\ref{newpc}).
If~$\Gamma$ is a set of sentences over~$\sigma_1$ then $gr(\Gamma)$ stands
for the set of formulas $gr(F)$ for all~$F$ in~$\Gamma$.\

The statement of the theorem
refers also to the system $\HT^\infty$ (Section~\ref{sec:infht})
extended by the axioms
$gr(\emph{Defs})$.  These axioms express the meaning of predicate
symbols~(\ref{newpc}).

\begin{theorem} \label{thm1}
  For any pure rule~$R$, $gr(\tau^*(R))$ is equivalent to $\tau(R)$
 in $\HTD$.
\end{theorem}

About {\sc mgc} programs~$\Pi_1$,~$\Pi_2$ we say that they are \emph{strongly
  equivalent} to each other if $\tau(\Pi_1)$ is strongly
equivalent to~$\tau(\Pi_2)$.  This condition guarantees that for any
{\sc mgc} program~$\Pi$ (and, more generally, for any logic program~$\Pi$ in
a similar language), $\Pi_1\cup\Pi$ has the same stable models as
$\Pi_2\cup\Pi$.

\begin{theorem}\label{thm2}
  {\sc mgc} programs~$\Pi_1$,~$\Pi_2$ are strongly equivalent to each other
  iff
  \beq
  gr(\tau^*(\Pi_1))\hbox{ is equivalent to }gr(\tau^*(\Pi_2))
  \hbox{ in }\HTD.
  \eeq{thm2a}
\end{theorem}

Thus the claim that MGC programs~$\Pi_1$, $\Pi_2$ are strongly equivalent to
each other can be always established, in principle, by deriving each of
the infinitary propositional formulas $gr(\tau^*(\Pi_1))$,
$gr(\tau^*(\Pi_2))$ from the other in $\HTD$.
Theorem~\ref{thm3} below shows that
in some cases such a claim can be justified by operating with finite
formulas---with first-order formulas of the signature~$\sigma_1$.  Instead
of $\HTD$ we can use the logic of here-and-there with arithmetic
\cite{lif21a} extended by the axiom schemas \emph{Defs\/}:

\begin{theorem} \label{thm3}
  For any {\sc mgc} programs~$\Pi_1$,~$\Pi_2$, if
  the formulas $\tau^*(\Pi_1)$ and $\tau^*(\Pi_2)$ are equivalent in
  $\HTAD$ then~$\Pi_1$ and $\Pi_2$ are strongly equivalent to each
  other.
\end{theorem}

As an example, we will use~$\HTD$ to verify that
rules~(\ref{g1}),~(\ref{g3}),~(\ref{g3a}) are strongly equivalent to each
other.  The translation~$\tau^*$ transforms these rules into the formulas
\beq
\forall Y\left(
  \exists C\left(C=Y\land\atl^{X;}_{\tau^B(p(X))}(C)\right)
  \land\tau^B(Y=\num 1),
\to q\right),
\eeq{g1f}
\beq
  \exists C\left(C=\num 1\land\atl^{X;}_{\tau^B(p(X))}(C)\right)
\to q,
\eeq{g3f}
\beq
\forall X(\tau^B(p(X))\to q).
\eeq{g3af}
The first two formulas are equivalent to each other
in intuitionistic predicate calculus with equality, which is a
subsystem of $\HTA$; this is clear from the fact that
$\tau^B(Y=\num 1)$ stands for the formula
$\exists Z_1Z_2(Z_1=Y\land Z_2=\num 1\land Z_1=Z_2)$, which is
intuitionistically equivalent to $Y=\num 1$.  Furthermore,~(\ref{g3f}) is
intuitionistically equivalent to
\beq
\atl^{X;}_{\tau^B(p(X))}(\num 1)
\to q.
\eeq{g3fb}
Using the axiom
$$\atl^{X;}_{\tau^B(p(X))}(\num 1)\lrar
  \exists X \tau^B(p(X))$$
of $\HTAD$,~(\ref{g3fb}) can be transformed into the formula
$\exists X\,\tau^B(p(X))\to q$,
which is intuitionistically equivalent to~(\ref{g3af}).

\section{Proofs} \label{sec:proofs}

In this section, the word ``equivalent'' in application to infinitary
propositional formulas refers to equivalence in $\HT^\infty$ whenever the
deductive system is not specified.

\begin{lemma} \label{prop1}
  For any tuple $\bft$ of terms in the language of mini-\gringo\ and any
  tuple~$\bf r$ of precomputed terms of the same length, the formula
  $gr(val_\bft({\bf r}))$ is provable in $\HT^\infty$ if
  ${\bf r}\in[\bft]$, and refutable otherwise.
\end{lemma}

\begin{proof}
  For the case when $\bft$ is a single term,
the assertion of the lemma can be proved by induction
\cite[Proposition~1]{lif19}.  The general case easily follows.
\end{proof}

The following fact is Proposition~2 by Lifschitz et al.~\citeyear{lif19}.

\begin{lemma}\label{prop2}
  If~$L$ is a ground literal or ground comparison in the language of
  mini-\gringo\ then $gr(\tau^B(L))$ is equivalent to~$\tau(L)$.
\end{lemma}

\begin{lemma} \label{newlemma}
  Let $\bfx$, $\bfv$, $\bfw$ be disjoint lists of distinct general variables,
  and let~$A$ be an aggregate atom
  $count\{\bfx:{\bf L}\} \prec t$
such that every variable occurring in~$\bf L$ belongs to one of these three
lists, and every variable occurring in~$t$ belongs to~$\bfv$.  For any
list $\bf v$ of precomputed terms of the same length as~$\bfv$, the
formula $\tau\left(A^\bfv_{\bf v}\right)$ is equivalent in $\HTD$ to 
\beq
gr\left(\exists C\left(val_{t^\bfv_{\bf v}}(C) \land
    Atleast^{\bfx;\bfv}_{\exists\bfw\tau^B({\bf L})}({\bf v},C)\right)\right)
\eeq{new1}
if~$\prec$ is $\geq$, and to
\beq
gr\left(\exists C\left(val_{t^\bfv_{\bf v}}(C) \land
    Atmost^{\bfx;\bfv}_{\exists\bfw\tau^B({\bf L})}({\bf v},C)\right)\right)
\eeq{new2}
if~$\prec$ is $\leq$.
\end{lemma}

\begin{proof}
\emph{Case~1:} $\prec$ is $\geq$.
Formula~(\ref{new1}) can be written as
$$\bigvee_c \left(gr(\val{t^\bfv_{\bf v}}c) \land
  \atl^{\bfx;\bfv}_{\exists\bfw\tau^B({\bf L})}({\bf v},c)\right),$$
where~$c$ ranges over precomputed terms.  From Lemma~\ref{prop1} we
see that it is equivalent to
$$\bigvee_{c\in\left[t^\bfv_{\bf v}\right]} 
  \atl^{\bfx;\bfv}_{\exists\bfw\tau^B({\bf L})}({\bf v},c).$$
Consider the infinitary formula obtained by grounding~(\ref{def1})
with $\exists\bfw\tau^B({\bf L})$ as~$F$.  One of its conjunctive terms is
$$\atl^{\bfx;\bfv}_{\exists\bfw\tau^B({\bf L})}({\bf v},c)\lrar
gr\left(\left(\exists_{\geq c} \bfx \exists\bfw\tau^B({\bf L})\right)^\bfv_{\bf v}\right).$$
Consequently~(\ref{new1}) is equivalent in $\HTD$ to
\beq
\bigvee_{c\in\left[t^\bfv_{\bf v}\right]} 
gr\left(\left(\exists_{\geq c} \bfx
    \exists\bfw\tau^B({\bf L})\right)^\bfv_{\bf v}\right).
\eeq{dis}

\emph{Case~1.1:} The set $\left[t^\bfv_{\bf v}\right]$ is empty.
Then~(\ref{dis}) is the
empty disjunction~$\bot$; $\tau\left(A^\bfv_{\bf v}\right)$ is~$\bot$ as
well.  \emph{Case~1.2:} The set $\left[t^\bfv_{\bf v}\right]$ is non-empty.
Since~$t$ is
interval-free, this set is a singleton~$\{c\}$, so that~(\ref{dis}) is
\beq
gr\left(\left(\exists_{\geq c} \bfx \exists\bfw\tau^B({\bf L})\right)^\bfv_{\bf v}\right)
\eeq{nondis}
and $\tau\left(A^\bfv_{\bf v}\right)$ is
\beq
\bigvee_{\Delta\,:\,\num{|\Delta|} \geq c}\;\bigwedge_{{\bf x}\in\Delta}\;\bigvee_{\bf w}
   \tau\!\left({\bf L}^{\bfx,\bfv,\bfw}_{\,{\bf x},\;{\bf v},\;{\bf w}}\right).
\eeq{tau1}
\emph{Case~1.2.1:} $c\leq\num 0$.  Then~(\ref{nondis}) is~$\top$.  The
disjunctive term of~(\ref{tau1}) with $\Delta=\emptyset$ is the empty
conjunction~$\top$, so that~(\ref{tau1}) is equivalent to~$\top$.
\emph{Case~1.2.2:} for all~$n$, $c>\num n$.  Then~(\ref{nondis}) is~$\bot$.
Formula~(\ref{tau1}) is the empty disjunction~$\bot$ as well.
\emph{Case~1.2.3:} $c$ is a numeral~$\num n$, $n>0$.  Then~(\ref{nondis}) is
$$gr\left(\exists\bfx_1\cdots\bfx_n\left(\bigwedge_{i=1}^n
    \exists\bfw\left(\tau^B({\bf L})^{\bfx,\bfv}_{\bfx_i,{\bf v}}\right)
    \land\bigwedge_{i<j}\neg(\bfx_i=\bfx_j)\right)\right),$$
This formula can be rewritten as
$$\bigvee_{{\bf x}_1,\dots,{\bf x}_n}\left(\bigwedge_{i=1}^n
  \bigvee_{\bf w}gr\left((\tau^B{\bf L})^{\bfx,\bfv,\bfw}_{{\bf x}_i,{\bf v},{\bf w}}\right)
  \land\bigwedge_{i<j}\neg gr({\bf x}_i={\bf x}_j)\right)$$
(${\bf x}_1,\dots{\bf x}_n,{\bf w}$ range over tuples of precomputed terms).
The part $\bigwedge_{i<j}\neg(gr({\bf x}_i={\bf x}_j))$
is equivalent to~$\top$ if the tuples ${\bf x}_1,\dots{\bf x}_n$
are pairwise distinct, that is to say, if the cardinality of the set
$\{{\bf x}_1,\dots{\bf x}_n\}$ is~$n$; otherwise this conjunction is
equivalent to~$\bot$.  Consequently (\ref{nondis}) is equivalent to
$$\bigvee_{\Delta\,:\,|\Delta|=n}\,\bigwedge_{{\bf x}\in\Delta}
\bigvee_{\bf w}gr
\left((\tau^B{\bf L})^{\bfx,\bfv,\bfw}_{{\bf x},\;{\bf v},\;{\bf w}}\right).$$
The formula
$gr\!\left((\tau^B{\bf L})^{\bfx,\bfv,\bfw}_{{\bf x},\;{\bf v},\;{\bf w}}\right)$
can be rewritten as
$gr\!\left(\tau^B\left({\bf L}^{\bfx,\bfv,\bfw}
    _{{\bf x},\;{\bf v},\;{\bf w}}\right)\right)$.
By Lemma~\ref{prop2}, it is equivalent to
$\tau\left({\bf L}^{\bfx,\bfv,\bfw}_{{\bf x},\;{\bf v},\;{\bf w}}\right)$,
so that~(\ref{nondis}) is equivalent to
\beq
\bigvee_{\Delta\,:\,|\Delta| = n}\;\bigwedge_{{\bf x}\in\Delta}\;\bigvee_{\bf w}
   \tau\!\left({\bf L}^{\bfx,\bfv,\bfw}_{\,{\bf x},\;{\bf v},\;{\bf w}}\right).
   \eeq{e1}
Disjunction~(\ref{tau1}) can be obtained from this disjunction by adding
similar disjunctive terms with sets~$\Delta$ containing more than~$n$
tuples.  Since each of these disjunctive terms is stronger than some of the
disjunctive terms in~(\ref{e1}), the two disjunctions are equivalent.

\emph{Case~2:} $\prec$ is $\leq$.
Formula~(\ref{new2}) is equivalent in $\HTD$ to
\beq
\bigvee_{c\in\left[t^\bfv_{\bf v}\right]} 
gr\left(\left(\exists_{\leq c} \bfx
    \exists\bfw\tau^B({\bf L})\right)^\bfv_{\bf v}\right);
\eeq{disa}
this is parallel to the argument in Case~1.

\emph{Case~2.1:} The set $\left[t^\bfv_{\bf v}\right]$ is empty.
Then~(\ref{disa}) is the
empty disjunction~$\bot$; $\tau\left(A^\bfv_{\bf v}\right)$ is~$\bot$ as
well.  \emph{Case~2.2:} The set $\left[t^\bfv_{\bf v}\right]$ is non-empty.
Since~$t$ is
interval-free, this set is a singleton~$\{c\}$, so that~(\ref{disa}) is
\beq
gr\left(\left(\exists_{\leq c} \bfx \exists\bfw\tau^B({\bf L})\right)^\bfv_{\bf v}\right)
\eeq{nondisa}
and $\tau\left(A^\bfv_{\bf v}\right)$ is
\beq
\bigwedge_{\Delta\,:\,\num{|\Delta|}>c}\;
\neg\bigwedge_{{\bf x}\in\Delta}\;\bigvee_{\bf w}
   \tau\!\left({\bf L}^{\bfx,\bfv,\bfw}_{\,{\bf x},\;{\bf v},\;{\bf w}}\right).
\eeq{tau1a}
\emph{Case~2.2.1:} $c<\num 0$.  Then~(\ref{nondisa}) is~$\bot$.  The
conjunctive term of~(\ref{tau1a}) with $\Delta=\emptyset$ is $\neg\top$,
so that~(\ref{tau1a}) is equivalent to~$\bot$.
\emph{Case~2.2.2:} for all~$n$, $c>\num n$.  Then~(\ref{nondisa}) is~$\top$.
Formula~(\ref{tau1a}) is the empty conjunction~$\top$ as well.
\emph{Case~2.2.3:} $c$ is a numeral~$\num n$, $n>0$.  Then~(\ref{nondisa}) is
$$gr\left(\forall\bfx_1\cdots\bfx_{n+1}\left(\bigwedge_{i=1}^{n+1}
    \exists\bfw\left(\tau^B({\bf L})^{\bfx,\bfv}_{\bfx_i,{\bf v}}\right)
    \to\bigvee_{i<j}\bfx_i=\bfx_j\right)\right).$$
This formula can be rewritten as
$$\bigwedge_{{\bf x}_1,\dots,{\bf x}_{n+1}}\left(\bigwedge_{i=1}^{n+1}
  \bigvee_{\bf w}gr\left((\tau^B{\bf L})^{\bfx,\bfv,\bfw}_{{\bf x}_i,{\bf v},{\bf w}}\right)
  \to\bigvee_{i<j} gr({\bf x}_i={\bf x}_j)\right).$$
The consequent $\bigvee_{i<j}gr({\bf x}_i={\bf x}_j))$
is equivalent to~$\bot$ if the tuples ${\bf x}_1,\dots{\bf x}_{n+1}$
are pairwise distinct, that is to say, if the cardinality of the set
$\{{\bf x}_1,\dots{\bf x}_{n+1}\}$ is~$n+1$; otherwise this conjunction is
equivalent to~$\top$.  Consequently (\ref{nondisa}) is equivalent to
$$\bigwedge_{\Delta\,:\,|\Delta|=n+1}\neg\bigwedge_{{\bf x}\in\Delta}
\bigvee_{\bf w}gr
\left((\tau^B{\bf L})^{\bfx,\bfv,\bfw}_{{\bf x},\;{\bf v},\;{\bf w}}\right)$$
and furthermore to
\beq
\bigwedge_{\Delta\,:\,|\Delta| = n+1}\neg\bigwedge_{{\bf x}\in\Delta}\;\bigvee_{\bf w}
   \tau\!\left({\bf L}^{\bfx,\bfv,\bfw}_{\,{\bf x},\;{\bf v},\;{\bf w}}\right);
   \eeq{e1a}
this is parallel to the argument in Case~1.2.3.
Conjunction~(\ref{tau1a}) can be obtained from this conjunction by adding
similar conjunctive terms with finite sets~$\Delta$ containing more than~$n+1$
tuples.  Since each of these conjunctive terms is weaker than some of the
conjunctive terms in~(\ref{e1a}), the two conjunctions are equivalent
to each other.
\end{proof}

% \pagebreak
\noindent\emph{Proof of Theorem~\ref{thm1}}

\noindent
Assume, for instance, that~$R$ is a basic rule
$p(\bft) \ar B_1\land\cdots\land B_n$;
for choice rules and constraints the proof is similar. Then~$\tau^*(R)$ is
$$\forall\bfv{\bf Z}
(B^*_1\land\cdots\land B^*_n\land\val{\bf t}{\bf Z}\to p({\bf Z})),$$
where $\bfv$ is the list of global variables of~$R$, and
$B_i^*$ are the formulas defined in Section~\ref{sec:trans}.  It
follows that $gr(\tau^*(R))$ is the conjunction of the formulas
$$gr\left((B^*_1)^\bfv_{\bf v}\right)\land\cdots\land
  gr\left((B^*_n)^\bfv_{\bf v}\right)
\land gr\left(\val{\bf t^\bfv_{\bf v}}{\bf r}\right)
\to p({\bf r})$$
over all tuples $\bf v$ of precomputed terms of the same length as~${\bf V}$
and all tuples $\bf r$ of precomputed terms of the same length as~${\bf Z}$.
By Lemma~\ref{prop1}, we can conclude that $gr(\tau^*(R))$ is equivalent
to the formula
$$\bigwedge_{\bf v}
\left(\left(gr\left((B^*_1)^\bfv_{\bf v}\right)\land\cdots\land
  gr\left((B^*_n)^\bfv_{\bf v}\right)\right)
\to \bigwedge_{{\bf r}\in \left[{\bf t}^\bfv_{\bf v}\right]} p({\bf r})\right).
$$
Each of the formulas
\beq
gr\left((B^*_i)^\bfv_{\bf v}\right)
\eeq{eqi}
$(i=1,\dots,n)$ is equivalent in $\HTD$ to
$\tau\left((B_i)^\bfv_{\bf v}\right)$.
Indeed, if~$B_i$ is a literal or a comparison then~(\ref{eqi})
is $gr\left(\left(\tau^B(B_i)\right)^\bfv_{\bf v}\right)$,
which can be also written as
$gr\left(\tau^B\left((B_i)^\bfv_{\bf v}\right)\right)$;
this formula is equivalent in $\HTD$ to
$\tau\left((B_i)^\bfv_{\bf v}\right)$
by Lemma~\ref{prop2}.
If~$B_i$ is an aggregate atom then~(\ref{eqi}) is equivalent in $\HTD$ to
$\tau\left((B_i)^\bfv_{\bf v}\right)$ by Lemma~\ref{newlemma}.
Consequently $gr(\tau^*(R))$ is equivalent in $\HTD$ to
\beq
\bigwedge_{\bf v}
\left(\tau\left((B_1)^\bfv_{\bf v}\right)\land\cdots\land
  \tau\left((B_n)^\bfv_{\bf v}\right)
\to \bigwedge_{{\bf r}\in \left[{\bf t}^\bfv_{\bf v}\right]} p({\bf r})\right).
\eeq{athere}
It remains to observe that instances of~$R$ are rules of the form
$$
p(\bft^\bfv_{\bf v}) \ar (B_1)^\bfv_{\bf v}\land\cdots\land (B_n)^\bfv_{\bf v},
$$
so that~(\ref{athere}) is~$\tau(R)$.

\begin{lemma}\label{conserv}
  If an infinitary propositional formula over the set of precomputed
  atoms is provable in $\HTD$ then it is provable in~$\HT^\infty$.
\end{lemma}

\noindent\emph{Proof (sketch)}

\noindent
The set $gr(\emph{Defs})$ consists of infinitary propositional
formulas of the forms
\beq
\bigwedge_{\bf v}\left(\atl^{\bfx;\bfv}_F({\bf v},r)\lrar
  gr\left(\exists_{\geq r} \bfx F^{\bf V}_{\bf v}\right)\right)
\eeq{def1gr}
and
\beq
\bigwedge_{\bf v}\left(\atm^{\bfx;\bfv}_F({\bf v},r)\lrar
  gr\left(\exists_{\leq r} \bfx F^{\bf V}_{\bf v}\right)\right).
\eeq{def2gr}
A derivation from $gr(\emph{Defs})$ can be visualized as a tree with axioms
of $\HT^\infty$ and formulas~(\ref{def1gr}),~(\ref{def2gr}), attached to
leaves.  In such a tree, modify all formulas by replacing
\begin{itemize}
\item
atoms $\atl^{\bfx;\bfv}_F({\bf v},r)$
  by $gr\left(\exists_{\leq r}\bfx F^{\bfv}_{\bf v}\right)$, and
\item
atoms $\atm^{\bfx;\bfv}_F({\bf v},r)$
  by $gr\left(\exists_{\geq r}\bfx F^{\bfv}_{\bf v}\right)$.
\end{itemize}
The result is a derivation from formulas that are provable in $\HT^\infty$.
If the formula attached to the root does not contain
$\atl^{\bfx;\bfv}_F$, $\atm^{\bfx;\bfv}_F$ then it is not affected by this
transformation.

\medskip\noindent\emph{Proof of Theorem~\ref{thm2}}

\noindent
By Theorem~\ref{thm1}, each of the equivalences
$$gr(\tau^*(\Pi_1)) \lrar \tau(\Pi_1),\ gr(\tau^*(\Pi_2)) \lrar \tau(\Pi_2)$$
is provable in $\HTD$.  Consequently
condition~(\ref{thm2a}) is equivalent to the condition
\beq
  \tau(\Pi_1)\hbox{ is equivalent to }\tau(\Pi_2)\hbox{ in }\HTD.
\eeq{thm2b}
By Lemma~\ref{conserv},~(\ref{thm2b}) is equivalent to the condition
$$\tau(\Pi_1)\lrar\tau(\Pi_2)\hbox{ is provable in }\HT^\infty,$$
which holds if and only if $\Pi_1$ is strongly equivalent to~$\Pi_2$.

\begin{lemma}\label{lemma-gr}
  If a sentence~$F$ over the signature~$\sigma_1$ is provable in~$\HTAD$ then
  $gr(F)$ is provable in~$\HTD$.
\end{lemma}

\noindent\emph{Proof (sketch)}

\noindent
For any axiom~$S$ of~$\HTA$, the formula
$gr(S)$ is provable in~$\HT^\infty$.  (To be precise,
axioms of~$\HTA$ are sequents, and the transformation~$gr$ needs to be
applied to the universal closure of the formula corresponding to~$S$.)
For any instance
$$\frac{S_1\ \cdots\ S_k}{S}$$
of an inference rule of~$\HTA$, the formula $gr(S)$ is derivable from
$gr(S_1),\dots,gr(S_k)$ in~$\HT^\infty$.
It follows that for any formula~$F$ that is derivable from \emph{Defs}
in~$\HTA$, the formula $gr(F)$ is derivable from
$gr(\emph{Defs})$ in~$\HT^\infty$.

\medskip\noindent\emph{Proof of Theorem~\ref{thm3}}

\noindent
By Lemma~\ref{lemma-gr}, if the equivalence
$\tau^*(\Pi_1)\lrar\tau^*(\Pi_2)$ is provable in~$\HTAD$ then the
equivalence
$gr(\tau^*(\Pi_1))\lrar gr(\tau^*(\Pi_2))$ is provable in~$\HTD$.  Then,
by Theorem~\ref{thm2}, the programs~$\Pi_1$ and~$\Pi_2$ are strongly
equivalent.

\section{Related Work}\label{sec:rw}

Fandinno et al.~\citeyear{fan22a} defined a translation similar
to~$\tau^*$ for an answer set programming language that is in some ways
less expressive than mini-\gringo\ with counting (no arithmetic
operations), and in some ways more expressive (the \verb|#sum| aggregate is
allowed, besides \verb|#count|).  The main difference between that
approach to transforming aggregate expressions into formulas and the one
described above is that the former employs function symbols
in the role that predicate symbols~(\ref{newpc}) play here.
As discussed in Footnote~\ref{ft1}, thinking of {\tt \#count} as a
function may be misleading.  This is apparently the reason why
the adequacy of the translation due to Fandinno et al.~is
only guaranteed for programs without
positive recursion through aggregates \cite[Theorem~3]{fan22a}.  This
assumption is not satisfied, for instance, for program~(\ref{g2}),~(\ref{bad}).

The technical problems discussed in this paper are specific for the approach
to aggregates implemented in the answer set solver {\sc clingo} and do not
appear in the same form, for instance, in the theory of the solver {\sc dlv}
\cite{fab11}.  The semantics of aggregates based on the vicious circle
principle \cite{gel19}, unlike the \clingo\ semantics,
makes rule~(\ref{g2}) strongly equivalent to each of the rules
(\ref{g1})--(\ref{g3a}).

\section*{Acknowledgements}

Many thanks to Jorge Fandinno, Michael Gelfond, Yuliya Lierler,
and the anonymous referees
for comments on preliminary versions of this paper.

% \bibliographystyle{acmtrans}
% \bibliography{bib}

\end{document}